\documentclass[a4paper, 11pt]{article}

\usepackage[hang,flushmargin]{footmisc}

\usepackage[dvipsnames]{xcolor}
\usepackage{graphicx}
\usepackage[a4paper,top=28mm,bottom=32mm,left=28mm,right=28mm]{geometry}

\usepackage{amsthm,amsmath,amsfonts,amssymb}
\usepackage{mathtools}

\usepackage{textpos}

\usepackage{caption,subcaption}
\usepackage{enumitem}

\usepackage{tikz, tikz-cd}

\usepackage{float}
\usepackage{braket}

\usepackage[hyperfootnotes=false,colorlinks,
linkcolor={blue!80!black},
citecolor={red!60!black},
urlcolor={blue!80!black}]{hyperref}

\numberwithin{equation}{section}

\newtheorem{theorem}{Theorem}[section]
\newtheorem{lemma}[theorem]{Lemma}
\newtheorem{proposition}[theorem]{Proposition}
\newtheorem{corollary}[theorem]{Corollary}

\theoremstyle{definition}
\newtheorem{definition}[theorem]{Definition}

\theoremstyle{remark}
\newtheorem*{remark}{Remark}

\newcommand{\ip}[2]{\big\langle #1 \,\big|\, #2 \big\rangle}

\begin{document}
	
\begin{center}
	{\boldmath\LARGE\bf Near-Tsirelson Bell--CHSH Violations in Quantum Field Theory via Carleman and Hankel Operators}
	
\vspace{2ex}
	
{\sc David Dudal\footnote{\noindent KU Leuven Kulak, Kortrijk, Belgium.  \href{mailto:david.dudal@kuleuven.be}{\texttt{david.dudal@kuleuven.be}.}} and Ken Vandermeersch\footnote{\noindent KU Leuven Kulak, Kortrijk, Belgium.  \href{mailto:ken.vandermeersch@kuleuven.be}{\texttt{ken.vandermeersch@kuleuven.be}}. Research supported by Methusalem grant METH/21/03 --- long-term structural funding of the Flemish government.}}
\end{center}

\vspace{-0.075cm}

\begin{abstract}
    \noindent We study Bell--Clauser--Horne--Shimony--Holt (Bell--CHSH) violations in the vacuum state of free spinor fields in $(1+1)$-dimensional Minkowski spacetime. We construct explicit smooth compactly supported test functions with spacelike separated supports whose Bell--CHSH correlators converge to Tsirelson’s bound $2\sqrt2$. In the massless case, after passage to the time-zero slice and a natural symmetry reduction, the problem reduces to the quadratic form of the Carleman operator on $L^2([0,\infty))$. Near-maximal Bell violation is then governed by the spectral edge $\pi$, and explicit near-extremizers are obtained from compactly supported cutoffs of the generalized eigenfunction $x^{-1/2}$. This also explains the appearance of the constant $\pi$ in earlier wavelet-based formulations. In the massive case, the same reduction leads to a Hankel operator with kernel $mK_1(m(x+y))$, where $K_1$ denotes the modified Bessel function of the second kind of order $1$, and exponentially damped variants of the massless test functions again yield Bell--CHSH values converging to $2\sqrt2$. Therefore, we establish a direct link between Bell--CHSH violations for free $(1+1)$-dimensional spinor fields and the spectral theory of Carleman and Hankel operators on the half-line.
\end{abstract}
	
\section{Introduction}
Bell’s inequality and its Clauser--Horne--Shimony--Holt (CHSH) form \cite{Bell1964,Clauser1969} quantify the tension between local realism and quantum correlations. In relativistic quantum field theory (QFT), one may ask whether Bell--CHSH violations can be realized by observables localized in spacelike separated regions. In the algebraic QFT framework, Summers and Werner showed that this is indeed the case for free Bose and Fermi quantum field theories \cite{SummersI1987,SummersII1987,Summers1987}: in the vacuum state, suitable spacelike separated observables yield Bell correlations arbitrarily close to Tsirelson’s bound \(2\sqrt2\) \cite{Cirelson80}. The involved smooth compactly supported test functions are, however, only defined implicitly. For recent discussions on this topic, see \cite{Guimaraes:2024mmp,Caribe:2026mam}.

The aim of this paper is more concrete. We consider free spinor fields in \((1+1)\)-dimensional Minkowski spacetime and give an explicit construction of smooth compactly supported test functions for Alice and Bob, with spacelike separated supports, whose Bell--CHSH correlator approaches \(2\sqrt2\). In contrast to the more general existence results of \cite{SummersI1987,SummersII1987,Summers1987}, the present approach is explicit and reduces the Bell problem in this free-field model to concrete spectral problems on the positive half-line: the Carleman operator in the massless case and a Hankel operator with Bessel kernel in the massive case. This yields a direct operator-theoretic explanation of the near-Tsirelson violations and makes it possible to write down explicit test functions.

\medskip

To formulate the Bell problem in this setting, let
\[
h(t,x)=\bigl(h_1(t,x),h_2(t,x)\bigr)^{\mathsf T}
\]
be a smooth compactly supported two-component spinor test function. Smearing the free spinor field with \(h\) gives a field operator \(\psi(h)\), and from this one forms the self-adjoint observable
\[
\mathcal A_h:=\psi(h)+\psi^\dagger(h).
\]
The vacuum two-point function then induces an inner product on test functions,
\begin{equation}\label{eq:intro-innerprod}
\ip{h}{h'}:=\bra{0}\mathcal A_h\mathcal A_{h'}\ket{ 0}.
\end{equation}
Given Alice's test functions \(f,f'\) and Bob's test functions \(g,g'\), the corresponding Bell--CHSH correlator is 
\begin{align*}
\Braket{\mathcal C}
&:=
 i \, \big( \ip{f}{g} + \ip{f'}{g} + \ip{f}{g'} - \ip{f'}{g'}  \big).
\end{align*}
The Bell problem is to find \(f,f'\) and \(g,g'\), with spacelike separated supports, such that \(\bigl|\Braket{\mathcal C}\bigr|\) approaches \(2\sqrt2\).
\medskip

The first step is to pass from the QFT formulation to a one-variable problem on the half-line. After localizing Alice’s test functions on the negative half-line and Bob’s on the positive half-line, we use temporal mollifiers to pass to the time-zero slice and obtain purely spatial inner-product formulas with kernels on \((-\infty,0]\times[0,\infty)\). This reduction is already present in the earlier physics literature; for completeness and to fix the notation used in the main text, Appendix~\ref{sec:app} gives a fully rigorous derivation from \eqref{eq:intro-innerprod}. A natural symmetry ansatz then reduces the four Bell pairings to a single one, thereby reducing the Bell problem to a one-function variational problem for an integral operator on \(L^2\big([0,\infty)\big)\).

\medskip

In the massless case, the relevant integral operator is the Carleman operator
\[
(\mathbf C\phi)(x)=\int_0^\infty \frac{\phi(y)}{x+y}\,\mathrm dy, \qquad \phi \in L^2\big( [0,\infty) \big).
\]
After the symmetry reduction, near-Tsirelson Bell--CHSH violation is governed by the spectral edge \(\pi\) of \(\mathbf C\). This yields an explicit family of smooth compactly supported, spacelike separated test functions whose Bell--CHSH values converge to \(2\sqrt 2\), obtained from compactly supported cutoffs of the generalized eigenfunction \(x \mapsto x^{-1/2}\).

This operator-theoretic viewpoint also resolves the wavelet-based approach of \cite{Dudal2023,DV26}. The numerical problem studied in \cite{Dudal2023} and the matrices studied in \cite{DV26} are precisely finite-dimensional compressions of the Carleman operator in a specific Haar wavelet basis, so the wavelet formulation is a basis-dependent version of the same spectral problem. In particular, the conjectural asymptotic value $\pi$ in \cite{DV26} is explained by the operator norm \(\big\lVert \mathbf C \big\rVert = \pi\), and our argument provides a direct proof of \cite[Conjectures~A \&~B]{DV26}.
 
\medskip 

This operator-theoretic reduction also extends to the massive case. The passage to the time-zero slice and symmetry reduction now leads to a Hankel operator with kernel \(mK_1\big(m(x+y)\big)\), where \(K_1\) is the modified Bessel function of the second kind of order $1$. The associated variational problem is again governed by the spectral edge~\(\pi\), and this again yields explicit test functions, obtained as exponentially damped variants of the massless family, whose Bell--CHSH values converge to \(2\sqrt2\).

\medskip

The main contribution of the paper is a constructive operator-theoretic reformulation of near-Tsirelson Bell--CHSH violations for free $(1+1)$-dimensional spinor fields, reducing the problem to the spectral theory of the Carleman operator in the massless case and of a Hankel operator with Bessel kernel in the massive case. We expect that this viewpoint will be useful for further extensions, in particular to the bosonic case.

\medskip 

The paper is organized as follows. In Section~\ref{sec:free} we treat the massless case and reduce the Bell problem to the Carleman operator. In Section~\ref{sec:massive} we treat the massive case and obtain the corresponding Hankel-operator formulation with Bessel kernel. We end with a short outlook in Section~\ref{sec:out}. Appendix~\ref{sec:app} contains a rigorous derivation of the purely spatial inner-product formulas used in Sections~\ref{sec:free} and~\ref{sec:massive}, making precise a reduction that was already used in the earlier physics literature.

\section{The massless case} \label{sec:free}
 In the massless case, the relevant integral kernel is \((y-x)^{-1}\), and the Bell problem can be reformulated as an extremal problem for the Carleman operator \(\mathbf C\) on \(L^2\big([0,\infty)\big)\). We first show that, after a symmetry reduction, it suffices to construct a single smooth function, compactly supported inside \((0,\infty)\), with Rayleigh quotient close to the spectral edge \(\pi\). We then give an explicit family of such functions and deduce near-Tsirelson Bell--CHSH violations. 

\subsection{Reduction to the Carleman operator}

We work with real-valued purely spatial test functions
\[
f,f' \in C_c^\infty\big( (-\infty,0); \mathbb R^2 \big), \qquad g,g'\in C_c^\infty \big( (0,\infty); \mathbb R^2 \big)
\] 
for Alice and Bob. By Corollary~\ref{cor:app-main-massless} in Appendix~\ref{sec:app}, the reduced purely spatial inner product and norms for \(a \in \{f,f'\}\) and \(b \in \{g,g'\}\) are given as follows:
\begin{align}\label{eq:massless-spatial-pairing}
\ip{a}{b}
&=
\frac{i}{\pi}
\iint_{(-\infty,0]\times [0,\infty)}
\frac{a_1(x)b_1(y)-a_2(x)b_2(y)}{y-x}\,\mathrm dx\,\mathrm dy, \\
\ip{a}{a}
&=
\int_{-\infty}^0 \bigl(a_1(x)^2+a_2(x)^2\bigr)\,\mathrm dx,
\qquad
\ip{b}{b}
=
\int_0^\infty \bigl(b_1(y)^2+b_2(y)^2\bigr)\,\mathrm dy.
\label{eq:massless-spatial-norm}
\end{align}

The problem is to construct normalized test functions \(f,f'\) and \(g,g'\) such that the absolute value of the correlator 
\[
\big\lvert \Braket{ \mathcal C }\big\rvert = \big\lvert \ip fg + \ip {f'}g + \ip f{g'} - \ip {f'}{g'}  \big\rvert
\]
is as close as possible to Tsirelson's bound \(2 \sqrt 2\). 

\medskip 

We impose the following symmetry ansatz: 
\begin{align}\label{eq:massless-ansatz}
    f_2'&=f_1, \qquad f_1'=-f_2, \qquad g_2'=g_1, \qquad g_1'=-g_2, \notag \\
    f_2 &= -cf_1, \qquad g_2=-cg_1, \qquad f_1(x) = -g_1(-x),
\end{align}
where \(c:= \sqrt 2-1\) so that \(c^2 + 2c -1 =0\). This symmetry ansatz is also consistent with the symmetries observed in the earlier numerical constructions \cite{Dudal2023,DV26}. The point of this ansatz is that it reduces the four Bell pairings to a single one. 

\begin{lemma}\label{lemma:massless-symmetry-reduction}
Let \((f,f'),(g,g')\) satisfy \eqref{eq:massless-ansatz}. Then
\begin{equation*}
\ip{f}{g}
=
\ip{f'}g
=
\ip f{g'}
=
-\ip{f'}{g'}, \qquad \ip ff = \ip {f'}{f'} = \ip gg = \ip {g'}{g'}.
\end{equation*}
Consequently,
\[
\big\lvert \Braket{ \mathcal C }\big\rvert
=
4\, \left\lvert \ip fg \right\rvert.
\]
\end{lemma}
\begin{proof}
The identities follow by direct substitution of \eqref{eq:massless-ansatz} into the massless inner product formulas \eqref{eq:massless-spatial-pairing}, \eqref{eq:massless-spatial-norm}.
\end{proof}

Thus it suffices to construct a single function \(g_1 \in C_c^\infty((0,\infty))\) for which \(\left|\ip{f}{g}\right|\) is as close as possible to \(\sqrt2/2\), subject to the normalization \(\ip{g}{g}=1\).

\medskip 

Consider the Hilbert space \(H:= L^2\big( [0,\infty) \big)\) with its standard inner product
\[
\Braket{\phi, \psi } := \int_0^\infty \overline{\phi(x)}\,\psi(x)\,\mathrm dx.
\]

\begin{definition}
The \emph{Carleman operator} \(\mathbf C \in \mathcal B(H)\) is defined by
\[
(\mathbf C\phi)(x)
:=
\int_0^\infty \frac{\phi(y)}{x+y}\, \mathrm dy,
\qquad
\phi \in H.
\]
\end{definition}

We recall the basic spectral facts that will be used below.

\begin{proposition}\label{prop:carleman-spectrum}
The Carleman operator \(\mathbf C\) is bounded, self-adjoint, and nonnegative on \(H\). Moreover, its spectrum \(\sigma(\mathbf C) = [0,\pi]\) is purely absolutely continuous with multiplicity \(2\). In particular,
\[
\big\lVert\mathbf C\big\rVert = \pi.
\]
\end{proposition}

\begin{proof}
See, for example, \cite[Chs.~1 and 10]{Pel03} and \cite{yafaev2014spectral}.
\end{proof}

The Bell pairing is now expressed directly in terms of the quadratic form \(\Braket{ g_1,\mathbf C g_1}\).

\begin{lemma}\label{lemma:massless-carleman-reduction}
Let \((f,f'),(g,g')\) satisfy \eqref{eq:massless-ansatz}. Then
\begin{equation*}
\ip{f}{g}
=
\frac{1-c^2}{i\pi}\,\Braket{ g_1,\mathbf C g_1}, \qquad \ip{g}{g}
=
(1+c^2)\,\big\lVert g_1\big\rVert^2.
\end{equation*}
\end{lemma}
\begin{proof} 
Substituting \(f_2=-cf_1\), \(g_2=-cg_1\), and \(f_1(x)=-g_1(-x)\) into \eqref{eq:massless-spatial-pairing}, and then changing variables \(x\mapsto -x\), we obtain
\[
\ip{f}{g}
=
\frac{1-c^2}{i\pi}
\int_0^\infty\int_0^\infty
\frac{g_1(x)g_1(y)}{x+y}\,\mathrm dx \mathrm dy
=
\frac{1-c^2}{i\pi}\,\Braket{ g_1,\mathbf C g_1}.
\]
Likewise,
\[
\ip{g}{g}
=
\int_0^\infty \bigl(g_1(x)^2+g_2(x)^2\bigr)\, \mathrm dx
=
(1+c^2) \, \big\lVert g_1\big\rVert^2.
\]
\end{proof}

The reduction is now complete: near-maximal Bell violation is governed by the spectral edge of the Carleman operator.

\begin{proposition}\label{prop:massless-reduction}
Suppose that \(\phi^{(\varepsilon)} \in C_c^\infty((0,\infty))\) is \(L^2\)-normalized and satisfies
\[
\Braket{ \phi^{(\varepsilon)},\mathbf C\phi^{(\varepsilon)}}
\longrightarrow \pi
\qquad
\text{as }\varepsilon\to0^+.
\]
Define
\[
g_1^{(\varepsilon)}
:=
\frac{1}{\sqrt{1+c^2}}\,\phi^{(\varepsilon)},
\qquad
g_2^{(\varepsilon)}
:=
-\frac{c}{\sqrt{1+c^2}}\,\phi^{(\varepsilon)},
\]
and define \(g'{}^{(\varepsilon)}, f^{(\varepsilon)}, f'{}^{(\varepsilon)}\) by \eqref{eq:massless-ansatz}. Then
\(
\ip{g^{(\varepsilon)}}{g^{(\varepsilon)}} = 1
\) 
and
\[
\bigl\lvert \ip{f^{(\varepsilon)}}{g^{(\varepsilon)}} \bigr\rvert
=
\frac{1-c^2}{\pi(1+c^2)}
\,
\Braket{ \phi^{(\varepsilon)},\mathbf C\phi^{(\varepsilon)}}
\longrightarrow
\frac{1-c^2}{1+c^2}
=
\frac{\sqrt2}{2}.
\]
In particular,
\[
\left| \Braket{ \mathcal C^{(\varepsilon)} } \right|
\longrightarrow
2\sqrt2.
\]
\end{proposition}

\begin{proof}
This is immediate from Lemmas~\ref{lemma:massless-symmetry-reduction} and \ref{lemma:massless-carleman-reduction}.
\end{proof}

\subsection{Explicit approximate eigenfunctions}

In view of Proposition~\ref{prop:massless-reduction}, it remains to construct a family of \(L^2\)-normalized functions in \(C_c^\infty((0,\infty))\) whose Carleman Rayleigh quotient \(\Braket{ \phi^{(\varepsilon)}, \mathbf{C}\phi^{(\varepsilon)} }\) converges to the spectral edge \(\pi\). An explicit choice that works is the following.

For each \(0<\varepsilon<1\), let \(\widetilde\phi^{(\varepsilon)} \in C_c^\infty((0,\infty))\) be given by
\begin{center}
\begin{minipage}[t]{0.6\textwidth}
\[
\widetilde \phi^{(\varepsilon)}(x) = \begin{cases}
    0 &  x\in \big[0,\frac \varepsilon 2\big]\\[0.3em]
    \tau^{(\varepsilon)}(2x-\varepsilon)\, x^{-1/2}
      & x\in \big( \frac \varepsilon 2,\varepsilon\big)\\[0.3em]
    x^{-1/2}
      & x\in \big[\varepsilon, \frac 1 \varepsilon \big]\\[0.3em]
    \tau^{(\varepsilon)}(2\varepsilon-\varepsilon^2 x)\, x^{-1/2}
      & x\in\big(\frac 1 \varepsilon,\frac 2 \varepsilon\big)\\[0.3em]
    0 & x \in \big[\frac2 \varepsilon, \infty\big)
\end{cases}\] 
\end{minipage} \begin{minipage}[t]{0.38\textwidth}
\vspace{30pt}
\centering
\begin{tikzpicture}[scale=0.7]

  \draw[red, very thick] (0.8,0) -- (1.5,0);
  \draw[red, very thick] (6.4,0) -- (6.8,0);

  \draw[red, very thick]
    plot[smooth] coordinates {
    (1.4,0) 
      (1.5,0.01)
      (1.52,0.020)
      (1.66,0.120)
      (1.80,0.320)
      (1.92,0.650)
      (2.02,1.020)
      (2.10,1.300)
      (2.16,1.420)
      (2.2,1.453)
      (2.24,1.45)
    };

  \draw[red, thin, dashed]
    plot[smooth] coordinates {
      (1.60,1.76)
      (1.72,1.69)
      (1.85,1.62)
      (2.00,1.54)
      (2.10,1.49)
      (2.20,1.46)
    };

  \draw[red, very thick]
    plot[smooth] coordinates {
      (2.2,1.46)
      (2.45,1.37)
      (2.75,1.28)
      (3.10,1.19)
      (3.50,1.11)
      (3.95,1.05)
      (4.45,1.00)
      (4.85,0.98)
      (5.2,0.96)
    };

  \draw[red, thin, dashed]
    plot[smooth] coordinates {
      (5.20,0.96)
      (5.45,0.94)
      (5.70,0.91)
      (5.95,0.88)
      (6.20,0.87)
      (6.45,0.85)
    };

  \draw[red, very thick]
    plot[smooth] coordinates {
      (5.2,0.96)
      (5.34,0.95)
      (5.50,0.93)
      (5.68,0.90)
      (5.84,0.84)
      (6.00,0.73)
      (6.14,0.56)
      (6.26,0.32)
      (6.34,0.12)
      (6.4,0)
    };

  \draw[dashed] (1.5,0) -- (1.5,0.50);
  \draw[dashed] (2.2,0) -- (2.2,1.46);
  \draw[dashed] (5.2,0) -- (5.2,0.96);
  \draw[dashed] (6.4,0) -- (6.4,0.40);

  \draw[->] (0.4,0) -- (7.0,0);
  \draw[->] (0.8,-0.2) -- (0.75,2.3);

  \node[red] at (4.15,1.5) {$\widetilde\phi^{(\varepsilon)}$};

  \draw (1.5,-0.05) node[below] {$\frac{\varepsilon}{2}$};
  \draw (2.2,-0.05) node[below] {$\varepsilon$};
  \draw (5.2,-0.05) node[below] {$\frac{1}{\varepsilon}$};
  \draw (6.4,-0.05) node[below] {$\frac{2}{\varepsilon}$};

\end{tikzpicture}
\end{minipage}
\end{center}
where 
\[
\tau^{(\varepsilon)}(x)
=
\left[
1+\exp\left(\frac{\varepsilon(2x-\varepsilon)}{x(x-\varepsilon)}\right)
\right]^{-1},
\qquad 0<x<\varepsilon.
\]
is a smooth version of the standard step function.  We then define \(\phi^{(\varepsilon)}\) to be the \(L^2\)-normalization
\[
\phi^{(\varepsilon)}
:=
\frac{\widetilde\phi^{(\varepsilon)}}{\big\lVert\widetilde\phi^{(\varepsilon)}\big\rVert}.
\]
By construction, \(\phi^{(\varepsilon)}\) is smooth and compactly supported inside \((0,\infty)\), vanishing identically on \([0,\varepsilon/2]\). In particular, Alice's and Bob's test functions remain supported on strictly separated half-lines, thereby preserving spacelike separation. 

\begin{remark}  The above choice is motivated by the fact that \(x\mapsto x^{-1/2}\) is a ``generalized eigenfunction'' at the spectral edge \(\pi\) of the Carleman operator; see, for example, \cite[\S 1.2]{yafaev2014spectral}, \cite{Yafaev2010commutator,petrov2019exact}\footnote{For completeness, since the spectrum of the Carleman operator \(\mathbf C\) has multiplicity \(2\), there is another generalized eigenfunction at the spectral edge \(\pi\), namely \(x\mapsto x^{-1/2}\log x\); see \cite{petrov2019exact}. We do not use it here.}. Thus suitable compactly supported cutoffs of \(x^{-1/2}\) should have Rayleigh quotient close to \(\pi\). The next lemma makes this precise.    
\end{remark}
\begin{lemma}\label{lemma:massless-approx-eigenfunctions}
The family \(\phi^{(\varepsilon)}\) satisfies
\[
\Braket{ \phi^{(\varepsilon)}, \mathbf C \phi^{(\varepsilon)}}
\longrightarrow \pi
\qquad
\text{as } \varepsilon \to 0^+.
\]
\end{lemma}
\begin{proof}
Define the Carleman quadratic form
\[
Q_{\mathbf C}(\phi)
:=
\Braket{ \phi,\mathbf C\phi}
=
\int_0^\infty \int_0^\infty \frac{\overline{\phi(x)}\phi(y)}{x+y}\,\mathrm dx\,\mathrm dy,
\qquad \phi \in H.
\]
We apply the change of variables \(x = e^s, y = e^t\), and define \(
\widetilde\psi^{(\varepsilon)}(s)
:=
e^{s/2}\widetilde\phi^{(\varepsilon)}\big(e^s\big)
\)  for all \(s \in \mathbb R\). 
Then
\[
\big\lVert\widetilde\phi^{(\varepsilon)}\big\rVert^2
=
\int_0^\infty \bigl|\widetilde\phi^{(\varepsilon)}(x)\bigr|^2\,\mathrm dx
=
\int_{-\infty}^\infty \bigl|\widetilde\psi^{(\varepsilon)}(s)\bigr|^2\,\mathrm ds
=
\big\lVert\widetilde\psi^{(\varepsilon)}\big\rVert^2.
\]
Hence, if we set
\(
\psi^{(\varepsilon)}(s)
:=
{\widetilde\psi^{(\varepsilon)}(s)}/{\big\lVert\widetilde\psi^{(\varepsilon)}\big\rVert}
=
e^{s/2}\phi^{(\varepsilon)}\big(e^s\big),
\)
then \(\big\lVert\psi^{(\varepsilon)}\big\rVert=1\).

\medskip 

A direct computation shows that
\[
Q_{\mathbf C}\big(\phi^{(\varepsilon)}\big)
=
\iint_{\mathbb R^2}
h(s-t)\,\psi^{(\varepsilon)}(s)\psi^{(\varepsilon)}(t)\,\mathrm ds\,\mathrm dt, \qquad 
h(r)
:=
\frac{1}{2\cosh(r/2)}.
\]
After the change of variables \(r=s-t\), this becomes
\[
Q_{\mathbf C}\big(\phi^{(\varepsilon)}\big)
=
\int_{-\infty}^\infty h(r)\,W^{(\varepsilon)}(r)\,\mathrm dr, \qquad 
W^{(\varepsilon)}(r)
:=
\int_{-\infty}^\infty
\psi^{(\varepsilon)}(t+r)\psi^{(\varepsilon)}(t)\,\mathrm dt.
\]
Since
\(
\int_{-\infty}^\infty h(r)\,\mathrm dr = \pi
\)
and, by Cauchy--Schwarz,
\[
0 \leq W^{(\varepsilon)}(r) \leq \big\lVert\psi^{(\varepsilon)}\big\rVert^2 = 1,
\]
it suffices by the dominated convergence theorem to prove that
\[
W^{(\varepsilon)}(r)\longrightarrow 1
\qquad
\text{for every fixed } r\in\mathbb R.
\]

\medskip 

Let
\(
T:=\log(1/\varepsilon).
\) 
Since \(\widetilde\phi^{(\varepsilon)}(x)=x^{-1/2}\) on \([\varepsilon,1/\varepsilon]\), we have
\[
\widetilde\psi^{(\varepsilon)}(s)=1
\qquad
\text{for } s\in[-T,T].
\]
Moreover, \(\widetilde\psi^{(\varepsilon)}\) is supported on
\[
\big[-(T+\log 2),\,T+\log 2 \big]
\]
and satisfies \(0\leq \widetilde\psi^{(\varepsilon)} \leq 1\). Therefore,
\[
2T \leq \big\lVert\widetilde\psi^{(\varepsilon)}\big\rVert^2 \leq 2T+2\log 2.
\]
Fix \(r\in\mathbb R\). For \(\varepsilon\) sufficiently small, we have \(2T>|r|\). On the overlap
\(
[-T,T]\cap[-T-r,T-r],
\)
one has
\[
\widetilde\psi^{(\varepsilon)}(t+r)\widetilde\psi^{(\varepsilon)}(t)=1.
\]
Hence
\[
W^{(\varepsilon)}(r)
=
\frac{1}{\big\lVert\widetilde\psi^{(\varepsilon)}\big\rVert^2}
\int_{-\infty}^\infty
\widetilde\psi^{(\varepsilon)}(t+r)\widetilde\psi^{(\varepsilon)}(t)\,\mathrm dt
\geq
\frac{2T-|r|}{2T+2\log 2}
\longrightarrow 1 \qquad \text{as }\varepsilon\to0^+.
\]
Since also \(W^{(\varepsilon)}(r)\leq 1\), we conclude that
\[
W^{(\varepsilon)}(r)\longrightarrow 1
\qquad
\text{for every fixed } r\in\mathbb R.
\]

The dominated convergence theorem therefore yields
\[
Q_{\mathbf C}\big(\phi^{(\varepsilon)}\big)\longrightarrow \int_{-\infty}^\infty h(r)\,\mathrm dr = \pi \qquad \text{as \(\varepsilon \to 0^+\)}.
\]
This proves the lemma.
\end{proof}

\subsection{Near-Tsirelson Bell--CHSH violations}

We now combine Proposition~\ref{prop:massless-reduction} and Lemma~\ref{lemma:massless-approx-eigenfunctions}.

For each \(0<\varepsilon<1\), define Bob's test functions \((g,g')\) by
\begin{align*}
g_1^{(\varepsilon)}(x) &:= \frac{1}{\sqrt{1+c^2}}\,\phi^{(\varepsilon)}(x), &
g_2^{(\varepsilon)}(x) &:= -\frac{c}{\sqrt{1+c^2}}\,\phi^{(\varepsilon)}(x), \\
g_1'{}^{(\varepsilon)}(x) &:= \frac{c}{\sqrt{1+c^2}}\,\phi^{(\varepsilon)}(x), &
g_2'{}^{(\varepsilon)}(x) &:= \frac{1}{\sqrt{1+c^2}}\,\phi^{(\varepsilon)}(x),
\end{align*}
for \(x\geq 0\), and Alice's test functions \((f,f')\) by
\begin{align*}
f_1^{(\varepsilon)}(x) &:= -\frac{1}{\sqrt{1+c^2}}\,\phi^{(\varepsilon)}(-x), &
f_2^{(\varepsilon)}(x) &:= \frac{c}{\sqrt{1+c^2}}\,\phi^{(\varepsilon)}(-x), \\
f_1'{}^{(\varepsilon)}(x) &:= -\frac{c}{\sqrt{1+c^2}}\,\phi^{(\varepsilon)}(-x), &
f_2'{}^{(\varepsilon)}(x) &:= -\frac{1}{\sqrt{1+c^2}}\,\phi^{(\varepsilon)}(-x),
\end{align*}
for \(x\leq 0\). By construction, these test functions satisfy the symmetry ansatz \eqref{eq:massless-ansatz}.

\begin{theorem}\label{thm:massless-near-tsirelson}
For every \(0<\varepsilon<1\), the test functions
\(
(f^{(\varepsilon)},f'{}^{(\varepsilon)})\), \((g^{(\varepsilon)},g'{}^{(\varepsilon)})
\) 
are smooth and compactly supported on strictly separated half-lines, satisfy
\[
\ip{f^{(\varepsilon)}}{f^{(\varepsilon)}}
=
\ip{f'{}^{(\varepsilon)}}{f'{}^{(\varepsilon)}}
=
\ip{g^{(\varepsilon)}}{g^{(\varepsilon)}}
=
\ip{g'{}^{(\varepsilon)}}{g'{}^{(\varepsilon)}}
=1,
\]
and yield Bell--CHSH violations
\[
\left| \Braket{ \mathcal C^{(\varepsilon)} } \right|
\longrightarrow 2\sqrt2
\qquad
\text{as }\varepsilon\to0^+.
\]
\end{theorem}
\begin{proof}
By Lemma~\ref{lemma:massless-approx-eigenfunctions}, the family \(\phi^{(\varepsilon)}\) satisfies the hypothesis of Proposition~\ref{prop:massless-reduction}. Hence
\[
\ip{g^{(\varepsilon)}}{g^{(\varepsilon)}}=1,
\qquad
\bigl\lvert \ip{f^{(\varepsilon)}}{g^{(\varepsilon)}} \bigr\rvert \longrightarrow \frac{\sqrt2}{2},
\qquad
\left| \Braket{ \mathcal C^{(\varepsilon)} } \right|  \longrightarrow 2\sqrt2.
\]
The remaining norm identities follow from Lemma~\ref{lemma:massless-symmetry-reduction}.
\end{proof}

\begin{remark} 
The choice \(c=\sqrt2-1\) is singled out by the identity \(1-c^2=2c\), which makes the four Bell pairings collapse as in Lemma~\ref{lemma:massless-symmetry-reduction} and yields the near-maximal value \(2\sqrt2\). For general \(c\ge 0\), the same construction gives
\[
\ip{f}{g}
= - \ip{f'}{g'} =
\frac{1-c^2}{i\pi}\,\Braket{g_1,\mathbf C g_1},
\qquad
\ip{f'}{g}
=
\ip{f}{g'}
=
\frac{2c}{i\pi}\,\Braket{g_1,\mathbf C g_1}.
\]
After normalization, this yields
\[
\left|\Braket{\mathcal C^{(\varepsilon)}}\right|
\longrightarrow
\frac{2(1+2c-c^2)}{1+c^2}
\]
whenever \(\Braket{\phi^{(\varepsilon)},\mathbf C\phi^{(\varepsilon)}}\to\pi\). In this way one obtains all limiting Bell--CHSH values in \([2,2\sqrt2]\), and in particular all violations in \((2,2\sqrt2)\), by varying \(c\in[0,\sqrt2-1]\).
\end{remark}

\subsection{Relation with the Haar wavelet basis approach}

The operator-theoretic viewpoint developed above also settles the Haar wavelet-based formulation of \cite{Dudal2023,DV26}. Indeed, the matrices considered in \cite{DV26} arise as finite-dimensional compressions of the Carleman operator \(\mathbf C\) in a specific Haar wavelet basis. Thus \cite[Conjectures~A and B]{DV26} become basis-dependent special cases of the present operator-theoretic approach: the conjectural asymptotic value \(\pi\) is explained by the fact that \(\pi=\lVert\mathbf C\rVert\). In particular, if \(A(N,K)\) denotes the matrix introduced in \cite{DV26}, then
\[
\big\lVert A(N,K)\big\rVert_2=\lambda_{\max}\bigl(A(N,K)\bigr)\longrightarrow \pi
\qquad
\text{as } N,K\to\infty.
\]

This is an instance of the following standard fact from functional analysis.

\begin{lemma}
Let \(H\) be a Hilbert space, let \(\mathbf C\in \mathcal B(H)\) be self-adjoint and nonnegative, and let \((e_n)_{n\geq 1}\) be an orthonormal basis of \(H\). If \(A_N=\big[\Braket{ e_i,\mathbf C e_j}\big]_{1\leq i,j\leq N}\), then
\[
\lVert A_N\rVert_2=\lambda_{\max}(A_N)\longrightarrow \lVert \mathbf C\rVert
\qquad
\text{as } N\to\infty.
\]
\end{lemma}
\begin{proof}
    Let \(\mathbf P_N\) be the orthogonal projection onto \(\operatorname{span}\{e_1,\ldots,e_N\}\). Then \(A_N\) is the matrix of \(\big(\mathbf P_N\, \mathbf C \, \mathbf P_N\big)|_{\mathbf{P}_NH}\), so \(
    \big\lVert A_N \big\rVert_2 = \big\lVert \mathbf P_N\, \mathbf C \, \mathbf P_N \big\rVert. 
    \)
    Since \(\lVert \mathbf P_N\rVert = 1\), it holds that \(\lVert \mathbf P_N\, \mathbf C \, \mathbf P_N \rVert  \leq \lVert \mathbf C \rVert \). On the other hand, for any unit vector \(x \in H\) with \(\lVert \mathbf C \rVert -\varepsilon < \langle x,\mathbf C x\rangle \), one has \(\langle \mathbf P_N x, \mathbf C \mathbf P_N x\rangle \to \langle x,\mathbf C x \rangle \), thus  \(\lVert \mathbf C \rVert -\varepsilon \leq \liminf_N \, \lVert \mathbf{P}_N \, \mathbf C \, \mathbf P_N \rVert \).  Hence \(\lVert A_N \rVert_2 \to \lVert \mathbf C \rVert\). 
\end{proof}

\section{The massive case}\label{sec:massive}

We now consider free spinor fields in \((1+1)\) dimensions with mass \(m>0\). As in the massless case, the Bell problem reduces to a one-function variational problem on \(L^2([0,\infty))\). The relevant operator is now a Hankel operator with Bessel kernel. We show that the corresponding quadratic form is still asymptotically governed by the spectral edge \(\pi\), and we obtain an explicit family of smooth compactly supported, spacelike separated test functions whose Bell--CHSH values again converge to Tsirelson's bound \(2\sqrt 2\).

\subsection{Reduction to a Hankel operator}  

We again work with real-valued purely spatial test functions
\[
f,f' \in C_c^\infty\big( (-\infty,0); \mathbb R^2 \big), \qquad g,g'\in C_c^\infty\big( (0,\infty); \mathbb R^2 \big)
\] 
for Alice and Bob. 
 By Proposition~\ref{prop:app-bessel-pairing} in Appendix~\ref{sec:app}, the reduced purely spatial inner product between \(a \in \{f,f'\}\) and \(b\in\{g,g'\}\) is given as follows: 
\begin{align}\label{eq:massive-spatial-pairing}
\ip{a}{b}
&=
\frac{i}{\pi}
\iint_{(-\infty,0]\times [0,\infty)}
mK_1\bigl(m(y-x)\bigr)\bigl(a_1(x)b_1(y)-a_2(x)b_2(y)\bigr)\,\mathrm dx\,\mathrm dy, \end{align}
where \(K_1\) denotes the modified Bessel function of the second kind of order \(1\), for instance characterized for \(u>0\) by 
\begin{equation}\label{bessel:def}
K_1(u):=\int_0^\infty e^{-u\cosh t}\cosh t\,\mathrm dt,
\end{equation}
see \cite[\S6.22(5)]{Wat95}. 
The norms are still local and satisfy
\begin{equation}\label{eq:massive-spatial-norm}
\ip{a}{a}
=
\int_{-\infty}^0 \bigl(a_1(x)^2+a_2(x)^2\bigr)\,\mathrm dx,
\qquad
\ip{b}{b}
=
\int_0^\infty \bigl(b_1(y)^2+b_2(y)^2\bigr)\,\mathrm dy.
\end{equation}

We again construct normalized test functions \(f,f'\) and \(g,g'\) such that the absolute value of the correlator 
\[
\big\lvert \Braket{ \mathcal C }\big\rvert = \big\lvert \ip fg + \ip {f'}g + \ip f{g'} - \ip {f'}{g'}  \big\rvert
\]
is as close as possible to Tsirelson's bound \(2 \sqrt 2\). 

We retain the symmetry ansatz \eqref{eq:massless-ansatz}. The massive inner product formula \eqref{eq:massive-spatial-pairing} then implies the following analogue of Lemma~\ref{lemma:massless-carleman-reduction}.

\begin{lemma}\label{lemma:massive-hankel-reduction}
Let \(m>0\), and let \((f,f'),(g,g')\) satisfy \eqref{eq:massless-ansatz}. Then
\[
\bigl\lvert \Braket{ \mathcal C } \bigr\rvert
=
4\,\bigl\lvert \ip{f}{g} \bigr\rvert, \qquad \ip{f}{g}
=
\frac{1-c^2}{i\pi}\,\Braket{ g_1,\mathbf K_m g_1},
\]
where \(\mathbf K_m \in \mathcal B(H)\) is the integral operator
\[
(\mathbf K_m \phi)(x)
:=
\int_0^\infty mK_1\bigl(m(x+y)\bigr)\phi(y)\,\mathrm dy,
\qquad
\phi\in H:=L^2\big([0,\infty)\big).
\]
Moreover,
\begin{equation*}
\ip{f}{f} = \ip{f'}{f'} = \ip{g'}{g'} = \ip{g}{g}
=
(1+c^2)\,\big\lVert g_1\big\rVert^2.
\end{equation*}
\end{lemma}

\begin{proof}
The identity \(\bigl\lvert \Braket{\mathcal C}\bigr\rvert = 4\,\lvert \ip{f}{g}\rvert\) is proved exactly as in Lemma~\ref{lemma:massless-symmetry-reduction}. The formulas for \(\ip f g\) and the norms follow from \eqref{eq:massive-spatial-pairing}, \eqref{eq:massive-spatial-norm} and the symmetry ansatz.
\end{proof}

It is convenient to reduce first to the case \(m=1\). Let \(\mathbf K:=\mathbf K_1\), and for \(m>0\) define the unitary dilation operator \(\mathbf U_m\colon H\to H\) by
\[
(\mathbf U_m\phi)(x):=m^{1/2}\phi(mx).
\]
Notice that \(\mathbf U_m\) preserves \(C_c^\infty \big( (0,\infty)\big)\). Moreover, one computes that
\begin{equation}\label{lemma:massive-unitary-equivalence}
\mathbf K_m = \mathbf U_m \mathbf K \mathbf U_m^{-1},
\end{equation}
hence \(\mathbf K_m\) and \(\mathbf K\) are unitarily equivalent. We record the spectral properties of \(\mathbf K_m\).

\begin{proposition}\label{prop:massive-spectrum}
For every \(m>0\), the Hankel operator \(\mathbf K_m\) is bounded, self-adjoint, and nonnegative on \(H\). Moreover, it has no singular continuous spectrum, its absolutely continuous part is unitarily equivalent to multiplication by \(\lambda\) on \(L^2([0,\pi])\) with multiplicity \(1\), and
\[
\big\lVert\mathbf K_m\big\rVert=\pi.
\]
\end{proposition}

\begin{proof}
By \eqref{lemma:massive-unitary-equivalence}, it suffices to consider \(m=1\). Write
\[
(\mathbf K\phi)(x)
=
\int_0^\infty K_1(x+y)\phi(y)\,\mathrm dy
=
\int_0^\infty \frac{k(x+y)}{x+y}\phi(y)\,\mathrm dy,
\qquad
k(u):=uK_1(u).
\]
Thus \(\mathbf K=H(k)\) in the notation of \cite{howland1992spectral}. Standard properties of the modified Bessel functions (see \cite[\S3.71, \S7.23]{Wat95}) show that \(k\) is real-valued, bounded, and \(C^\infty\) on \((0,\infty)\), that
\begin{equation*}
\lim_{u\to0^+}k(u) =\lim_{u\to0^+} \big(1 + O(u^2|\log u|)\big) =1, \qquad K_1(u) \sim \left(\frac{\pi}{2u}\right)^{1/2} e^{-u} \qquad \text{as }u \to \infty,\end{equation*}
so \(\lim_{u\to\infty}k(u)=   0\). Similar considerations give 
\[
k''(u) = k(u) - K_0(u) = O(|\log u| ) \quad \text{as } u\to 0^+, \qquad k''(u) = O\big( u^{1/2} e^{-u} \big)\quad \text{as } u\to \infty,
\]
hence \(k''\) satisfies the hypotheses \((2.1)\)--\((2.2)\) of \cite[Th.~2]{howland1992spectral}. Hence \cite[Prop.~1.1 \& Th.~2]{howland1992spectral} imply that \(\mathbf K\) is bounded, has no singular continuous spectrum, and that its absolutely continuous part is unitarily equivalent to multiplication by \(\lambda\) on \(L^2([0,\pi])\) with multiplicity \(1\). Since the kernel \(K_1(x+y)\) is real and symmetric, \(\mathbf K\) is self-adjoint.

To prove nonnegativity, use the integral representation \eqref{bessel:def}, apply the substitution \(s = \cosh t\), and deduce that 
\[
\Braket{ \phi,\mathbf K\phi}
=
\int_1^\infty \frac{t}{\sqrt{t^2-1}}
\left|
\int_0^\infty e^{-tx}\phi(x)\,\mathrm dx
\right|^2
\mathrm dt
\ge 0, \qquad \phi\in C_c^\infty((0,\infty)). 
\]
By density of \(C_c^\infty((0,\infty))\) in \(H\) and boundedness of \(\mathbf K\), this extends to all \(\phi\in H\). Thus \(\mathbf K\ge 0\). Since the absolutely continuous part contains \([0,\pi]\), we obtain \(\lVert\mathbf K\rVert \ge \pi\).

On the other hand, the function \(k(u)=uK_1(u)\) satisfies
\[
k'(u)=\frac{\mathrm d}{\mathrm du}\bigl(uK_1(u)\bigr)=-uK_0(u)<0, \qquad u > 0,
\]
and \(\lim_{u\to0^+}k(u)=1\). Thus
\(
0<K_1(u)\le 1 /u \) for \(u > 0\).  
Therefore, for every \(\phi\in H\),
\[
0\le \Braket{ \phi,\mathbf K\phi}
\le
\int_0^\infty\int_0^\infty \frac{|\phi(x)|\,|\phi(y)|}{x+y}\,\mathrm dx\,\mathrm dy
=
\Braket{ |\phi|,\mathbf C|\phi|}
\le
\pi\big\lVert \phi\big\rVert^2,
\]
by Proposition~\ref{prop:carleman-spectrum}. Thus \(\lVert \mathbf K\rVert\le \pi\), and so \(\lVert \mathbf K\rVert =\pi\).
\end{proof}

As in the massless case, near-maximal Bell violation reduces to the construction of normalized functions whose \(\mathbf K_m\)-quadratic form is close to \(\pi\).

\begin{proposition}\label{prop:massive-reduction}
Let \(m>0\). Suppose that \(\Phi_m^{(\varepsilon)} \in C_c^\infty((0,\infty))\) is \(L^2\)-normalized and satisfies
\[
\Braket{ \Phi_m^{(\varepsilon)},\mathbf K_m \Phi_m^{(\varepsilon)}}
\longrightarrow \pi
\qquad
\text{as }\varepsilon\to0^+.
\]
Define
\[
g_1^{(\varepsilon)}
:=
\frac{1}{\sqrt{1+c^2}}\,\Phi_m^{(\varepsilon)},
\qquad
g_2^{(\varepsilon)}
:=
-\frac{c}{\sqrt{1+c^2}}\,\Phi_m^{(\varepsilon)},
\]
where \(c = \sqrt 2 -1\), and define \(g'{}^{(\varepsilon)},f^{(\varepsilon)},f'{}^{(\varepsilon)}\) by \eqref{eq:massless-ansatz}. Then
\[
\ip{f^{(\varepsilon)}}{f^{(\varepsilon)}}=\ip{f'^{(\varepsilon)}}{f'^{(\varepsilon)}}=\ip{g'^{(\varepsilon)}}{g'^{(\varepsilon)}}=\ip{g^{(\varepsilon)}}{g^{(\varepsilon)}}=1
\]
and
\[
\bigl\lvert \ip{f^{(\varepsilon)}}{g^{(\varepsilon)}} \bigr\rvert
=
\frac{1-c^2}{\pi(1+c^2)}
\Braket{ \Phi_m^{(\varepsilon)},\mathbf K_m \Phi_m^{(\varepsilon)}}
\longrightarrow
\frac{1-c^2}{1+c^2}
=
\frac{\sqrt2}{2}.
\]
In particular,
\[
\left\lvert \Braket{ \mathcal C^{(\varepsilon)} } \right\rvert
\longrightarrow
2\sqrt2.
\]
\end{proposition}

\begin{proof}
This is immediate from Lemma~\ref{lemma:massive-hankel-reduction}.
\end{proof}

\subsection{Exponentially damped massive family}

We now construct approximate eigenfunctions for \(\mathbf K\) directly from the massless family of Section~\ref{sec:free}. For \(0<\varepsilon<1\), let \(\widetilde\phi^{(\varepsilon)}\) be the unnormalized cutoff function introduced in Section~\ref{sec:free}, and define 
\begin{center}
\begin{minipage}[t]{0.45\textwidth}
\vspace{-2.5cm}
\begin{eqnarray*}
\widetilde\Phi^{(\varepsilon)}(x)&:=&e^{-x}\widetilde\phi^{(\varepsilon)}(x),\\
\Phi^{(\varepsilon)}&:=&\frac{\widetilde\Phi^{(\varepsilon)}}{\big\lVert \widetilde\Phi^{(\varepsilon)}\big\rVert}.
\end{eqnarray*}
\end{minipage} \begin{minipage}[t]{0.35\textwidth}
\centering
\begin{tikzpicture}[scale=0.7]

  \draw[red, very thick] (0.8,0) -- (1.5,0);
  \draw[red, very thick] (6.4,0) -- (6.8,0);

  \draw[red, very thick]
    plot[smooth] coordinates {
      (1.40,0)
      (1.50,0.01)
      (1.54,0.018)
      (1.66,0.090)
      (1.80,0.220)
      (1.92,0.430)
      (2.02,0.720)
      (2.10,0.950)
      (2.16,1.060)
      (2.20,1.10)
    };

  \draw[red, very thick]
    plot[smooth] coordinates {
      (2.20,1.10)
      (2.40,0.92)
      (2.65,0.72)
      (2.95,0.54)
      (3.25,0.40)
      (3.55,0.29)
      (3.90,0.21)
      (4.25,0.15)
      (4.60,0.11)
      (4.90,0.085)
      (5.20,0.065)
    };

  \draw[red, very thick]
    plot[smooth] coordinates {
      (5.20,0.065)
      (5.40,0.055)
      (5.60,0.043)
      (5.80,0.031)
      (6.00,0.020)
      (6.18,0.010)
      (6.30,0.004)
      (6.40,0)
    };

  \draw[dashed] (1.5,0) -- (1.5,0.45);
  \draw[dashed] (2.2,0) -- (2.2,1.10);
  \draw[dashed] (5.2,0) -- (5.2,0.20);
  \draw[dashed] (6.4,0) -- (6.4,0.20);

  \draw[->] (0.4,0) -- (7.0,0);
  \draw[->] (0.8,-0.2) -- (0.75,2.3);

  \node[red] at (4.15,1.25) {$\widetilde\Phi^{(\varepsilon)}$};

  \draw (1.5,-0.05) node[below] {$\frac{\varepsilon}{2}$};
  \draw (2.2,-0.05) node[below] {$\varepsilon$};
  \draw (5.2,-0.05) node[below] {$\frac{1}{\varepsilon}$};
  \draw (6.4,-0.05) node[below] {$\frac{2}{\varepsilon}$};

\end{tikzpicture}
\end{minipage}
\end{center}
Then \(\Phi^{(\varepsilon)}\in C_c^\infty((0,\infty))\), \(\big\lVert\Phi^{(\varepsilon)}\big\rVert=1\), and \(\Phi^{(\varepsilon)}\) vanishes identically on \([0,\varepsilon/2]\).

The reason for this choice is the bounds
\begin{equation}\label{eq:boundsK}
\frac{e^{-u}}{u}\leq K_1(u)\leq \frac{1}{u},
\qquad
u>0,
\end{equation}
where the lower one easily follows from the \(K_1\) integral representation \eqref{bessel:def} after applying the substitution \(s = \cosh t\). These bounds show that the quadratic form of \(\mathbf K\) is squeezed between the standard and an exponentially weighted Carleman-type form.

\begin{lemma}\label{lemma:massive-approx-eigenfunctions}
The family \(\Phi^{(\varepsilon)}\) satisfies
\[
\Braket{ \Phi^{(\varepsilon)},\mathbf K \Phi^{(\varepsilon)}}
\longrightarrow \pi
\qquad
\text{as } \varepsilon\to0^+.
\]
\end{lemma}

\begin{proof} We adapt the proof of Lemma~\ref{lemma:massless-approx-eigenfunctions}. The argument is the same after the logarithmic change of variables, except for the additional exponential weight induced by the bounds on \(K_1\). 
Set
\( 
N_\varepsilon:=\bigl\|\widetilde\Phi^{(\varepsilon)}\bigr\|^2.
\) 
Since
\(
\widetilde\Phi^{(\varepsilon)}(x)=e^{-x}\widetilde\phi^{(\varepsilon)}(x),
\)
the bounds \eqref{eq:boundsK}
imply
\begin{equation} \label{eq:massive-sandwich}
I_{2}^{(\varepsilon)} 
\le
\Braket{\Phi^{(\varepsilon)},\mathbf K\Phi^{(\varepsilon)}}
\le
I_{1}^{(\varepsilon)}, 
\end{equation}
where for \(\beta \in \{1,2\}\), we define the integral \(I_{\beta}^{(\varepsilon)}:=\frac{1}{N_\varepsilon}
\int_{0}^\infty \int_0^\infty \frac{e^{-\beta(x+y)}}{x+y}\,
\widetilde\phi^{(\varepsilon)}(x)\widetilde\phi^{(\varepsilon)}(y)\,
\mathrm dx\,\mathrm dy \). 
We argue that both bounds in \eqref{eq:massive-sandwich} converge to \(\pi\). As in the proof of Lemma~\ref{lemma:massless-approx-eigenfunctions}, put 
\[
T:=\log(1/\varepsilon),
\qquad
\widetilde\psi^{(\varepsilon)}(s):=e^{s/2}\widetilde\phi^{(\varepsilon)}(e^s),
\]
and for \(\beta \in \{1,2\}\) define
\[
\Psi_\beta^{(\varepsilon)}(s):=e^{-\beta e^s}\widetilde\psi^{(\varepsilon)}(s),
\qquad
W_\beta^{(\varepsilon)}(r):=
\int_{-\infty}^\infty
\Psi_\beta^{(\varepsilon)}(t+r)\Psi_\beta^{(\varepsilon)}(t)\,\mathrm dt.
\]
Then
\(
N_\varepsilon=\bigl\|\Psi_1^{(\varepsilon)}\bigr\|^2.
\) 

Since \(0\le \widetilde\psi^{(\varepsilon)}\le 1\), since \(\widetilde\psi^{(\varepsilon)}=1\) on \([-T,T]\), and since \(\widetilde\psi^{(\varepsilon)}\) is supported on
\[
 [-(T+\log 2),\,T+\log 2],
\]
the same overlap argument as in the massless case, now with the weight \(e^{-\beta e^s}\), implies that 
\[ 
N_\varepsilon=T+O(1), \qquad  W_\beta^{(\varepsilon)}(r)=T+O(1) \qquad \text{as }\varepsilon \to 0^+,\]
for fixed \(\beta \in \{1,2\}, r \in \mathbb{R}\). 
Hence
\[
\frac{W_\beta^{(\varepsilon)}(r)}{N_\varepsilon}\longrightarrow 1
\qquad
\text{for every fixed } r\in\mathbb R.
\]

After the change of variables \(x=e^s\), \(y=e^t\), the bounds in \eqref{eq:massive-sandwich} become
\[
I_1^{(\varepsilon)} =\int_{-\infty}^\infty h(r)\,\frac{W_1^{(\varepsilon)}(r)}{N_\varepsilon}\,\mathrm dr,
\qquad I_2^{(\varepsilon)} =\int_{-\infty}^\infty h(r)\,\frac{W_2^{(\varepsilon)}(r)}{N_\varepsilon}\,\mathrm dr,
\qquad
h(r):=\frac{1}{2\cosh(r/2)}.
\]
Moreover,
\(
0\le W_2^{(\varepsilon)}(r)\le W_1^{(\varepsilon)}(r)\le N_\varepsilon,
\) 
so dominated convergence yields
\[
\int_{-\infty}^\infty h(r)\,\frac{W_\beta^{(\varepsilon)}(r)}{N_\varepsilon}\,\mathrm dr
\longrightarrow
\int_{-\infty}^\infty h(r)\,\mathrm dr
=
\pi
\]
for $\beta\in\{1,2\}$. Applying \eqref{eq:massive-sandwich} yields the desired conclusion. 
\end{proof}

\subsection{Near-Tsirelson Bell--CHSH violations}

We now transfer the family from Lemma~\ref{lemma:massive-approx-eigenfunctions} to arbitrary mass \(m>0\). Define
\[
\Phi_m^{(\varepsilon)}:=\mathbf U_m\Phi^{(\varepsilon)}.
\]
Then \(\big\lVert \Phi_m^{(\varepsilon)}\big\rVert=1\), and \eqref{lemma:massive-unitary-equivalence} gives
\[
\Braket{ \Phi_m^{(\varepsilon)},\mathbf K_m \Phi_m^{(\varepsilon)}}
=
\Braket{ \Phi^{(\varepsilon)},\mathbf K \Phi^{(\varepsilon)}}
\longrightarrow \pi
\qquad
\text{as } \varepsilon\to0^+.
\]

For fixed \(m>0\), define Bob's test functions \((g,g')\) by
\begin{align*}
g_1^{(\varepsilon)}(x) &:= \frac{1}{\sqrt{1+c^2}}\,\Phi_m^{(\varepsilon)}(x), &
g_2^{(\varepsilon)}(x) &:= -\frac{c}{\sqrt{1+c^2}}\,\Phi_m^{(\varepsilon)}(x), \\
g_1'{}^{(\varepsilon)}(x) &:= \frac{c}{\sqrt{1+c^2}}\,\Phi_m^{(\varepsilon)}(x), &
g_2'{}^{(\varepsilon)}(x) &:= \frac{1}{\sqrt{1+c^2}}\,\Phi_m^{(\varepsilon)}(x),
\end{align*}
for \(x\geq 0\), and Alice's test functions \((f,f')\) by
\begin{align*}
f_1^{(\varepsilon)}(x) &:= -\frac{1}{\sqrt{1+c^2}}\,\Phi_m^{(\varepsilon)}(-x), &
f_2^{(\varepsilon)}(x) &:= \frac{c}{\sqrt{1+c^2}}\,\Phi_m^{(\varepsilon)}(-x), \\
f_1'{}^{(\varepsilon)}(x) &:= -\frac{c}{\sqrt{1+c^2}}\,\Phi_m^{(\varepsilon)}(-x), &
f_2'{}^{(\varepsilon)}(x) &:= -\frac{1}{\sqrt{1+c^2}}\,\Phi_m^{(\varepsilon)}(-x),
\end{align*}
for \(x\leq 0\).

\begin{theorem}\label{thm:massive-near-tsirelson}
For every \(m>0\), the test functions
\(
(f^{(\varepsilon)},f'{}^{(\varepsilon)}),\) \((g^{(\varepsilon)},g'{}^{(\varepsilon)})
\) are smooth and compactly supported on strictly separated half-lines, satisfy
\[
\ip{f^{(\varepsilon)}}{f^{(\varepsilon)}}
=
\ip{f'{}^{(\varepsilon)}}{f'{}^{(\varepsilon)}}
=
\ip{g^{(\varepsilon)}}{g^{(\varepsilon)}}
=
\ip{g'{}^{(\varepsilon)}}{g'{}^{(\varepsilon)}}
=1,
\]
and yield Bell--CHSH violations
\[
\left| \Braket{ \mathcal C^{(\varepsilon)} } \right|
\longrightarrow
2\sqrt2
\qquad
\text{as } \varepsilon\to0^+.
\]
\end{theorem}

\begin{proof}
By the displayed convergence above, the family \(\Phi_m^{(\varepsilon)}\) satisfies the hypothesis of Proposition~\ref{prop:massive-reduction}. The conclusion follows.
\end{proof}

\section{Outlook}\label{sec:out}

The present paper gives a constructive reformulation of the Bell--CHSH problem for free \((1+1)\)-dimensional spinor fields in terms of integral operators on the half-line. The general analysis of \cite{SummersI1987,SummersII1987,Summers1987} is formulated in the algebraic QFT framework and proves maximal Bell correlations for free Bose and Fermi fields, whereas the present approach is less general but more explicit: the Bell problem is reduced to concrete quadratic forms for the Carleman operator in the massless case and for a Hankel operator with Bessel kernel in the massive case.

A first natural question is whether this operator-theoretic reduction can be related more directly to the wedge-modular framework underlying \cite{SummersI1987,SummersII1987,Summers1987}. In the present setting, the size of the Bell violation is controlled by the spectral edge of the relevant integral operator, and explicit test functions can be constructed accordingly.

A second question is whether one can obtain an equally explicit construction for free bosonic fields. The existence of maximal Bell correlations is already known from \cite{SummersI1987,SummersII1987,Summers1987}, but in the bosonic case an explicit and comparably canonical realization in terms of Bell operators seems less clear. It would be interesting to know whether the problem can again be reduced to a tractable integral-operator formulation.

Finally, one may ask whether some version of the present method survives beyond the free case. At present this is speculative. A natural possibility is that the relevant two-point kernel should be expressed in terms of the positive K\"all\'en--Lehmann spectral density associated with the interacting two-point function \cite{Glimm:1987ng}. Whether this again leads to a tractable integral-operator problem, and whether such an operator retains enough structure for a useful spectral analysis, remains open.

\section*{Acknowledgments}
 We are grateful to S.P.~Sorella for some useful correspondence.

\appendix

\section{From the free-field inner product to the spatial kernels}\label{sec:app}

In this Appendix we derive the spatial inner-product formulas used in Sections~\ref{sec:free} and~\ref{sec:massive} from the free-field Bell setup.

\subsection{The momentum-space inner product}

We work with free two-component spinors in \((1+1)\)-dimensional Minkowski spacetime, with metric \(g=\mathrm{diag}(1,-1)\), and with Dirac matrices
\[
\gamma^0=\sigma_x = \begin{bmatrix}
    0 & 1 \\ 1 & 0
\end{bmatrix},
\qquad
\gamma^1=i\sigma_y = \begin{bmatrix}
    0 & 1 \\ -1 & 0
\end{bmatrix}.
\]
We denote
\(
\omega_k:=\sqrt{k^2+m^2}\) for \(m\ge 0\).

Let
\(
h(t,x)=\bigl(h_1(t,x),h_2(t,x)\bigr)^{\mathsf T}\) and
\(h'(t,x)=\bigl(h'_1(t,x),h'_2(t,x)\bigr)^{\mathsf T}
\)
be two real-valued test functions in \(C_c^\infty(\mathbb R^{1+1};\mathbb R^2)\). Define their on-shell Fourier transforms by
\[
\widehat{h_j}(k):=\int_{\mathbb R^2} h_j(t,x)e^{-i\omega_k t}e^{ikx}\,\mathrm dt\,\mathrm dx, \qquad \widehat{h'_j}(k):=\int_{\mathbb R^2} h'_j(t,x)e^{-i\omega_k t}e^{ikx}\,\mathrm dt\,\mathrm dx.
\]

With the above definition of Fourier transform, and under the assumptions \(\overline{\widehat{h_j}(k)}=\widehat{h_j}(-k)\) and \(\overline{\widehat{h'_j}(k)}=\widehat{h'_j}(-k)\), the article \cite{Dudal2023} rewrites the free-field Bell inner product~\eqref{eq:intro-innerprod} as a sum of a local term and a nonlocal term
\begin{align}\label{eq:app-momentum-pairing}
\Braket{ h\mid h'}
= I_1(h,h')+I_2(h,h'), 
\end{align}
where
\begin{align*}
I_1(h,h')
&:=
\int_{\mathbb R}\frac{\mathrm dk}{2\pi}
\Big(
\overline{\widehat{h_1}(k)}\,\widehat{h'_1}(k)
+
\overline{\widehat{h_2}(k)}\,\widehat{h'_2}(k)
\Big),
\\
I_2(h,h')
&:=
\int_{\mathbb R}\frac{\mathrm dk}{2\pi}\,
\frac{k}{\omega_k}
\Big(
\overline{\widehat{h_1}(k)}\,\widehat{h'_1}(k)
-
\overline{\widehat{h_2}(k)}\,\widehat{h'_2}(k)
\Big).
\end{align*}

\subsection{Temporal mollifiers and the zero-time limit}

Fix a nonnegative, real-valued, even mollifier
\[
\beta\in C_c^\infty((-1,1)),
\qquad
\beta(t)=\beta(-t),
\qquad
\int_{\mathbb R}\beta(t)\,\mathrm dt=1,
\]
and define \(\beta_\eta(t):=\eta^{-1}\beta(t/\eta)\) for \(\eta>0\).

If \(u=(u_1,u_2)^{\mathsf T}\in C_c^\infty(\mathbb R;\mathbb R^2)\) is a purely spatial test function, define its \emph{temporal mollification} by
\[
h_\eta[u](t,x):=\beta_\eta(t)\,u(x).
\]
Its on-shell Fourier transform factorizes as
\(
\widehat{(h_\eta[u])_j}(k)
=
\widehat{\beta_\eta}(\omega_k)\,\widehat{u_j}(k),
\)
where
\[
\widehat{\beta_\eta}(\omega):=\int_{\mathbb R}\beta_\eta(t)e^{-i\omega t}\,\mathrm dt,
\qquad
\widehat{u_j}(k):=\int_{\mathbb R}u_j(x)e^{ikx}\,\mathrm dx.
\]

The next lemma identifies the \(\eta\to0^+\) limit of the local term \(I_1\) and reduces the nonlocal term \(I_2\) in terms of a spatial multiplier \(k/\omega_k\).

\begin{lemma}\label{lemma:app-zero-time-limit}
Let \(u,v\in C_c^\infty(\mathbb R;\mathbb R^2)\) be purely spatial test functions. Then
\begin{align}
I_1\!\bigl(h_\eta[u],h_\eta[v]\bigr)
&\longrightarrow
\int_{\mathbb R}\bigl(u_1(x)v_1(x)+u_2(x)v_2(x)\bigr)\,\mathrm dx,
\label{eq:app-I1-limit}
\\
I_2\!\bigl(h_\eta[u],h_\eta[v]\bigr)
&\longrightarrow
\int_{\mathbb R}\frac{\mathrm dk}{2\pi}\,
\frac{k}{\omega_k}
\Big(
\widehat{u_1}(-k)\widehat{v_1}(k)
-
\widehat{u_2}(-k)\widehat{v_2}(k)
\Big),
\label{eq:app-I2-limit}
\end{align}
 as \(\eta \to 0^+\).
\end{lemma}

\begin{proof}
Since \(u_j\) is real-valued, \(\overline{\widehat{u_j}(k)}=\widehat{u_j}(-k)\). Since \(\beta\) is real-valued and even, \(\widehat{\beta_\eta}\) is real-valued and even, so
\[
\overline{\widehat{(h_\eta[u])_j}(k)}
=
\widehat{\beta_\eta}(\omega_k)\,\widehat{u_j}(-k)
=
\widehat{(h_\eta[u])_j}(-k).
\]
Thus the symmetry assumption used in \eqref{eq:app-momentum-pairing} holds for \(h_\eta[u]\). Moreover,
\(
\widehat{\beta_\eta}(\omega)=\widehat{\beta}(\eta\omega),
\) 
hence
\[
0\le \widehat{\beta_\eta}(\omega_k)^2\le \|\beta\|_{L^1}^2=1,
\qquad
\widehat{\beta_\eta}(\omega_k)^2\longrightarrow 1
\qquad
\text{as }\eta\to0^+,
\]
for every fixed \(k\in\mathbb R\).

Using the formulas for \(I_1\) and \(I_2\), together with the factorization above, we obtain
\begin{align*}
I_1\!\bigl(h_\eta[u],h_\eta[v]\bigr)
&=
\int_{\mathbb R}\frac{\mathrm dk}{2\pi}\,
\widehat{\beta_\eta}(\omega_k)^2
\Big(
\widehat{u_1}(-k)\widehat{v_1}(k)
+
\widehat{u_2}(-k)\widehat{v_2}(k)
\Big),
\\
I_2\!\bigl(h_\eta[u],h_\eta[v]\bigr)
&=
\int_{\mathbb R}\frac{\mathrm dk}{2\pi}\,
\widehat{\beta_\eta}(\omega_k)^2\,\frac{k}{\omega_k}
\Big(
\widehat{u_1}(-k)\widehat{v_1}(k)
-
\widehat{u_2}(-k)\widehat{v_2}(k)
\Big).
\end{align*}
Since \(u_j\) and \(v_j\) are smooth and compactly supported, their Fourier transforms are Schwartz functions. Hence both integrands are absolutely integrable and uniformly dominated in \(\eta\). Dominated convergence gives \eqref{eq:app-I2-limit}, while for \(I_1\) it yields
\[
I_1\!\bigl(h_\eta[u],h_\eta[v]\bigr)
\to
\int_{\mathbb R}\frac{\mathrm dk}{2\pi}
\Big(
\widehat{u_1}(-k)\widehat{v_1}(k)
+
\widehat{u_2}(-k)\widehat{v_2}(k)
\Big).
\]
The latter is exactly \eqref{eq:app-I1-limit} by Fourier inversion.
\end{proof}

\begin{corollary}\label{cor:app-norm}
Let \(u\in C_c^\infty(\mathbb R;\mathbb R^2)\) be real-valued. Then
\[
\Braket{ h_\eta[u]\mid h_\eta[u]}
\to
\int_{\mathbb R}\bigl(u_1(x)^2+u_2(x)^2\bigr)\,\mathrm dx
\qquad
\text{as }\eta\to0^+.
\]
\end{corollary}

\begin{proof}
If \(v=u\), the limit \eqref{eq:app-I2-limit} vanishes, because
\(
\widehat{u_j}(-k)\widehat{u_j}(k)=|\widehat{u_j}(k)|^2
\)
is even in \(k\), whereas \(k/\omega_k\) is odd.
\end{proof}

In the case of spatially separated supports, the limiting nonlocal term  \eqref{eq:app-I2-limit} can be rewritten in configuration space by means of a Bessel kernel. Below, \(K_\nu\) is the modified Bessel function of the second kind of order \(\nu\), see \cite{Wat95}. 

\begin{proposition}\label{prop:app-bessel-pairing}
Assume \(m>0\). Let \(u,v\in C_c^\infty(\mathbb R;\mathbb R^2)\) be purely spatial test functions and assume \[
\operatorname{supp}(u) \subset (-\infty, -\varepsilon/2), \qquad \operatorname{supp}(v) \subset (\varepsilon/2, \infty)
\] 
for some fixed \(\varepsilon > 0\). 
Then
\begin{align*} 
I_2\!\bigl(h_\eta[u],h_\eta[v]\bigr)
\longrightarrow
\frac{i}{\pi}
\iint_{\mathbb R^2} \!
mK_1\big(m(y-x)\big) \bigl(u_1(x)v_1(y)-u_2(x)v_2(y)\bigr)\,\mathrm dx\,\mathrm dy  
\end{align*}
as \(\eta\to0^+\).
\end{proposition}

\begin{proof} 
By Lemma~\ref{lemma:app-zero-time-limit}, it suffices to treat one pair \(a,b\in C_c^\infty(\mathbb R;\mathbb R)\), with
\(\operatorname{supp}(a)\subset(-\infty,-\varepsilon/2),\) \(\operatorname{supp}(b)\subset(\varepsilon/2,\infty),
\) 
and to show that
\begin{equation*}
\int_{\mathbb R}\frac{\mathrm dk}{2\pi}\,\frac{k}{\omega_k}\,\widehat a(-k)\widehat b(k)
=
\frac{i}{\pi}
\iint_{\mathbb R^2}
mK_1(m|x-y|)\,a(x)b(y)\,\mathrm dx\,\mathrm dy.
\end{equation*}

Set
\(
G(x):=\frac{1}{\pi}K_0(m|x|).
\) 
By \cite[\S13.21(8)]{Wat95}, one has \(\int_0^\infty K_0(t)\,\mathrm dt=\pi/2\), hence \(G\in L^1(\mathbb R)\). Moreover, using the cosine-transform identity \cite[\S6.22(14)]{Wat95} for \(K_0\), we obtain
\[
\widehat G(k)
=
\int_{\mathbb R}G(x)e^{ikx}\,\mathrm dx
=
\frac{2}{\pi}\int_0^\infty K_0(mx)\cos(kx)\,\mathrm dx
=
\frac{1}{\sqrt{k^2+m^2}}
=
\frac{1}{\omega_k}.
\]
Therefore
\[
(G*b)(x)
=
\int_{\mathbb R}\frac{\mathrm dk}{2\pi}\,\frac{1}{\omega_k}\,\widehat b(k)e^{-ikx}.
\]
Since \(k e^{-ikx}= i\,\frac{\mathrm d}{\mathrm dx}e^{-ikx}\), we obtain for \(x\in \operatorname{supp}(a)\),
\[
\int_{\mathbb R}\frac{\mathrm dk}{2\pi}\,\frac{k}{\omega_k}\,\widehat b(k)e^{-ikx}
=
i\,\frac{\mathrm d}{\mathrm dx}(G*b)(x)
=
i\int_{\mathbb R}G'(x-y)b(y)\,\mathrm dy,
\]
where the last identity is justified by the fact that \(|x-y|\ge \varepsilon\) for \(x\in \operatorname{supp}(a)\) and \(y\in \operatorname{supp}(b)\). 
Since
\[
G'(r)
=
\frac{1}{\pi}\frac{\mathrm d}{\mathrm dr}K_0(m|r|)
=
-\frac{1}{\pi}\operatorname{sgn}(r)\,mK_1(m|r|)
=
\frac{1}{\pi}\operatorname{sgn}(-r)\,mK_1(m|r|),
\]
and \(\operatorname{sgn}(y-x) =1\) for \(x \in \operatorname{supp}(a), y \in \operatorname{supp}(b)\), we obtain
\[
\int_{\mathbb R}\frac{\mathrm dk}{2\pi}\,\frac{k}{\omega_k}\,\widehat b(k)e^{-ikx}
=
\frac{i}{\pi}\int_{\mathbb R}mK_1\big(m(y-x)\big)\,b(y)\,\mathrm dy, \qquad x\in \operatorname{supp}(a).
\]
Multiplying by \(a(x)\) and integrating over \(x\) gives
\begin{align*}
\int_{\mathbb R}\frac{\mathrm dk}{2\pi}\,\frac{k}{\omega_k}\,\widehat a(-k)\widehat b(k)
&=
\int_{\mathbb R} a(x)
\left(
\int_{\mathbb R}\frac{\mathrm dk}{2\pi}\,\frac{k}{\omega_k}\,\widehat b(k)e^{-ikx}
\right)\mathrm dx
\\
&=
\frac{i}{\pi}
\iint_{\mathbb R^2}
mK_1\big(m(y-x)\big)\,a(x)b(y)\,\mathrm dx\,\mathrm dy.
\end{align*}
This proves the proposition. 
\end{proof}

\subsection{The inner products used in the main text}\label{subsec:innerprod}

Let
\(
f,f',g,g'\in C_c^\infty(\mathbb R;\mathbb R^2)
\)
be spatial test functions such that for some \(\varepsilon>0\),
\[
\operatorname{supp}(f),\operatorname{supp}(f')\subset(-\infty,-\varepsilon/2),
\qquad
\operatorname{supp}(g),\operatorname{supp}(g')\subset(\varepsilon /2,\infty).
\]
For \(0<\eta<\varepsilon/2\), define the corresponding spacetime test functions by temporal mollification:
\[
h_\eta[f](t,x):=\beta_\eta(t)f(x),
\qquad
h_\eta[g](t,x):=\beta_\eta(t)g(x),
\]
and similarly for \(f'\) and \(g'\). Then every Alice test function \(h_\eta[a]\), with \(a\in\{f,f'\}\), is strictly spacelike separated from every Bob test function \(h_\eta[b]\), with \(b\in\{g,g'\}\), since
\(
|x-y|\ge \varepsilon > 2\eta \ge |t-s|
\) 
on the corresponding supports. Accordingly, define
\[
\ip{a}{b}:=\lim_{\eta\to0^+}\Braket{h_\eta[a]\mid h_\eta[b]},
\qquad
a,b\in\{f,f',g,g'\}.
\]
\begin{theorem}\label{thm:app-main-massive}
Let \(a \in \{f,f'\}\) and \(b \in \{g,g'\}\) be as above.  Then: 
\begin{align}
\ip{a}{a}
&=
\int_{-\infty}^0 \bigl(a_1(x)^2+a_2(x)^2\bigr)\,\mathrm dx,
\qquad
\ip{b}{b}
=
\int_0^\infty \bigl(b_1(y)^2+b_2(y)^2\bigr)\,\mathrm dy,
\label{eq:app-main-norm}
\\
\ip{a}{b}
&=
\frac{i}{\pi}
\iint_{(-\infty,0]\times[0,\infty)}
mK_1\bigl(m(y-x)\bigr)
\bigl(a_1(x)b_1(y)-a_2(x)b_2(y)\bigr)\,\mathrm dx\,\mathrm dy.
\label{eq:app-main-massive-pairing}
\end{align}
\end{theorem}
\begin{proof}
    This follows from Corollary~\ref{cor:app-norm} and Proposition~\ref{prop:app-bessel-pairing}. 
\end{proof}

\begin{corollary}\label{cor:app-main-massless}
In the massless limit \(m \to 0^+\), the norm formula \eqref{eq:app-main-norm} is unchanged, and
\[
\ip{a}{b}
=
\frac{i}{\pi}
\iint_{(-\infty,0]\times[0,\infty)}
\frac{a_1(x)b_1(y)-a_2(x)b_2(y)}{y-x}\,\mathrm dx\,\mathrm dy.
\]
\end{corollary}
\begin{proof}
On the support of the integrand one has \(y-x\ge \varepsilon\), and (by e.g.\ the proof of Proposition~\ref{prop:massive-spectrum}) for every fixed \(r>0\) we have \(
0<mK_1(mr)\le \frac1r\) and \(mK_1(mr)\longrightarrow \frac1r\) as \(m \to 0^+\).  
Hence the kernel in \eqref{eq:app-main-massive-pairing} converges pointwise to \((y-x)^{-1}\) and is dominated by the integrable kernel \((y-x)^{-1}\), so the result follows by dominated convergence.
\end{proof}

\noindent \textbf{Conclusion.}
Theorem~\ref{thm:app-main-massive} and Corollary~\ref{cor:app-main-massless} identify the spatial pairings used in the main text as limits of genuine QFT pairings between spacelike separated spacetime test functions. For small \(\varepsilon>0\), the constructions in Sections~\ref{sec:free} and~\ref{sec:massive} produce real-valued purely spatial test functions
\(
f^{(\varepsilon)},f'{}^{(\varepsilon)},g^{(\varepsilon)},g'{}^{(\varepsilon)}
\)
with
\[
\operatorname{supp}\big(f^{(\varepsilon)}\big),\,\operatorname{supp}\big(f'{}^{(\varepsilon)}\big)\subset(-\infty,-\varepsilon/2),
\qquad
\operatorname{supp}\big(g^{(\varepsilon)}\big),\operatorname{supp}\big(g'^{(\varepsilon)}\big)\subset(\varepsilon/2,\infty),
\]
such that the corresponding spatial correlator \(\bigl|\Braket{\mathcal C^{(\varepsilon)}}\bigr|\) is arbitrarily close to \(2\sqrt2\). For \(0<\eta<\varepsilon/2\), the temporally mollified spacetime test functions remain strictly spacelike separated, and the corresponding QFT pairings converge to the spatial pairings as \(\eta\to0^+\). Hence, given any \(\delta>0\), one may first choose \(\varepsilon>0\) such that
\[
2\sqrt2-\delta<\bigl|\big\langle {\mathcal C^{(\varepsilon)}} \big\rangle\bigr|\le 2\sqrt2,
\]
and then choose \(0<\eta<\varepsilon/2\) sufficiently small so that the resulting QFT Bell--CHSH correlator also lies in \((2\sqrt2-\delta,\,2\sqrt2]\).


\begin{thebibliography}{10}\setlength{\itemsep}{-0.5mm} \setlength{\parsep}{0mm} \footnotesize 

\bibitem{Bell1964}
J.~S. Bell.
\newblock {On the Einstein Podolsky Rosen paradox}.
\newblock {\em Physics Physique Fizika}, 1:195--200, 1964.

\bibitem{Caribe:2026mam}
J.~G.~A. Carib{\'e}, M.~S. Guimaraes, I.~Roditi, and S.~P. Sorella.
\newblock {Modular Theory and the Bell-CHSH inequality in relativistic scalar Quantum Field Theory}, 2026. arXiv:{\href{https://arxiv.org/abs/2603.25873}{2603.25873}}. 

\bibitem{Clauser1969}
J.~F. Clauser, M.~A. Horne, A.~Shimony, and R.~A. Holt.
\newblock {Proposed Experiment to Test Local Hidden-Variable Theories}.
\newblock {\em Phys. Rev. Lett.}, 23:880--884, 1969.

\bibitem{Dudal2023}
D.~Dudal, P.~De~Fabritiis, M.~S. Guimaraes, I.~Roditi, and S.~P. Sorella.
\newblock {Maximal violation of the Bell-Clauser-Horne-Shimony-Holt inequality
  via bumpified Haar wavelets}.
\newblock {\em Phys. Rev. D}, 108:L081701, 2023.

\bibitem{DV26}
D.~Dudal and K.~Vandermeersch.
\newblock {Further evidence for near-Tsirelson Bell--CHSH violations in quantum
  field theory via Haar wavelets}. 
\newblock {\em Eur. Phys. J. C}, 86:349, 2026.

\bibitem{Glimm:1987ng}
J.~Glimm and A.~Jaffe.
\newblock {\em Quantum Physics: A Functional Integral Point of View}.
\newblock Springer-Verlag, New York, second edition, 1987.

\bibitem{Guimaraes:2024mmp}
M.~S. Guimaraes, I.~Roditi, and S.~P. Sorella.
\newblock {Bell{\textquoteright}s inequality in relativistic Quantum Field
  Theory}.
\newblock {\em Rev. Phys.}, 13:100121, 2025.

\bibitem{howland1992spectral}
J.~S. Howland.
\newblock Spectral theory of operators of {H}ankel type. {II}.
\newblock {\em Indiana Univ. Math. J.}, 41(2):409--426, 427--434, 1992.

\bibitem{Pel03}
V.~V. Peller.
\newblock {\em Hankel operators and their applications}.
\newblock Springer Monographs in Mathematics. Springer-Verlag, New York, 2003.

\bibitem{petrov2019exact}
V.~Petrov.
\newblock {On exact solutions of Hankel equations}.
\newblock {\em {St. Petersburg Math. J.}}, 30(1):123--148, 2019.

\bibitem{SummersI1987}
S.~J. Summers and R.~Werner.
\newblock {Bell’s inequalities and quantum field theory. I. General setting}.
\newblock {\em J. Math. Phys.}, 28(10):2440--2447, 1987.

\bibitem{SummersII1987}
S.~J. Summers and R.~Werner.
\newblock {Bell’s inequalities and quantum field theory. II. Bell’s
  inequalities are maximally violated in the vacuum}.
\newblock {\em J. Math. Phys.}, 28(10):2448--2456, 1987.

\bibitem{Summers1987}
S.~J. Summers and R.~Werner.
\newblock {Maximal violation of Bell's inequalities is generic in quantum field
  theory}.
\newblock {\em Comm. Math. Phys.}, 110:247--259, 1987.

\bibitem{Cirelson80}
B.~S. Tsirelson.
\newblock {Quantum generalizations of Bell's inequality}.
\newblock {\em Lett. Math. Phys.}, 4(2):93--100, 1980.

\bibitem{Wat95}
G.~N. Watson.
\newblock {\em A {T}reatise on the {T}heory of {B}essel {F}unctions}.
\newblock Cambridge University Press, Cambridge, 1995.

\bibitem{yafaev2014spectral}
D.~R. Yafaev.
\newblock {Spectral and scattering theory for perturbations of the Carleman
  operator}.
\newblock {\em St. Petersburg Math. J.}, 25(2):339--359, 2014.

\bibitem{Yafaev2010commutator}
D.~R. Yafaev.
\newblock {A commutator method for the diagonalization of Hankel operators}.
\newblock {\em {Funct. Anal. Appl.}}, 44(4):295--306,
  2010.

\end{thebibliography}
\end{document}